\documentclass[12pt]{article}

\usepackage{booktabs}       
\usepackage{nicefrac}       
\usepackage{microtype}      
\usepackage{amsfonts,amstext,amsmath,amssymb,amsthm}
\usepackage{smile}
\usepackage{algorithm}
\usepackage{algorithmic}
\usepackage{wrapfig}
\usepackage{lipsum}
\usepackage{ragged2e}
\usepackage{geometry}
\usepackage{setspace} 
\usepackage{tcolorbox}
\usepackage[symbol]{footmisc}
\renewcommand*{\thefootnote}{\fnsymbol{footnote}}

\usepackage[title,toc]{appendix}
\usepackage{minitoc}
\noptcrule

\usepackage{natbib}
\usepackage[utf8]{inputenc} 
\usepackage[T1]{fontenc}    
\usepackage{hyperref}   
\hypersetup{
	colorlinks=true,       
	linkcolor=blue,        
	citecolor=blue,        
	filecolor=magenta,     
	urlcolor=blue
}
\usepackage{url}            
\usepackage{booktabs}       
\usepackage{amsfonts}       
\usepackage{nicefrac}       
\usepackage{microtype}      
\usepackage[table,xcdraw]{xcolor}
\usepackage{algorithm}
\usepackage{algorithmic}
\usepackage[]{graphicx,float,latexsym}
\usepackage{amsfonts,amstext,amsmath,amssymb,amsthm}
\usepackage{wrapfig}
\usepackage{lipsum}
\usepackage{gensymb}
\usepackage{quoting,xparse}

\usepackage{xcolor}

\def\mytitle{{\it More Skills, Worse Agents?} Skill Shadowing \\ Degrades Performance When Expanding Skill Libraries}

\title{\mytitle}
\author{\\ Hongwen Song, \ Song (Vinson) Wei\footnote{Author correspondence to Song Wei, e-mail: \texttt{vinson.wei@databricks.com}}\\
\\
  \small{Databricks Inc.}
}
\date{\vspace{-20pt}}

\begin{document}

\maketitle

\vfill

\begin{abstract}
Skill libraries allow LLM agents to load task-specific instructions on demand, letting non-expert users solve domain-specific tasks through natural language without knowing which skills exist or how they work. However, performance degrades as libraries grow---by up to 21\% when scaling from a small set of helpful skills to a 202-skill library.
In this work, we formulate this performance degradation as the pass rate drop between loading a library of known-helpful skills and the full library. Moreover, we propose to decompose the pass rate drop by conditioning on the skill(s) invocation---which skills the agent selects during a trajectory---into two effects: \emph{skill shadowing}, where the agent selects wrong skills more often as the library expands, and \emph{context overhead}, where the enlarged context degrades execution even when selection is correct. We derive upper bounds on both effects to characterize their magnitudes of impacts to the pass rate drop. Our empirical estimates of the effects and their upper bounds both show that the \emph{skill shadowing} effect grows with library size and significantly contributes to the performance degradation, whereas the \emph{context overhead} effect remains small and indistinguishable from zero. This observed asymmetry establishes that the skill selection failure, not the enlarged context, is the primary bottleneck when expanding the skill libraries.
\end{abstract}


\vfill

\doparttoc 
\faketableofcontents 

\part{} 


\renewcommand*{\thefootnote}{\arabic{footnote}}

\newpage

\section{Introduction}\label{sec:intro}


LLM agents are increasingly extended not only by adding tool calls but also
by loading \emph{skills}---self-contained packages of instructions, scripts,
and reference documents that inject domain expertise into the agent's
context~\citep{xu2026agentskills}. The abstraction has rapidly become an
industry standard: Anthropic shipped Agent Skills for Claude and released
the specification as an open format, with partner-built skills from
Atlassian, Stripe, and Figma entering a curated
directory~\citep{anthropic2025skills}; OpenAI adopted the same
\texttt{SKILL.md} convention for Codex shortly
after~\citep{openai2025codexskills}.
The appeal is compositional: practitioners assemble dozens of skills into a
single agent, and foundational systems demonstrated that capability
libraries can be composed and selected at scale through language-based
interfaces---Voyager~\citep{wang2023voyager} incrementally builds a skill
library retrieved by embedding similarity, while
HuggingGPT~\citep{shen2023hugginggpt} selects among thousands of expert
models using their natural-language descriptions. The premise is
that a user can describe a task in plain language (``check why this
customer's cluster is slow'') and the agent autonomously selects the right
skill, without the user knowing which skills exist or how they work.
In our experiments, a 202-skill library cuts this premise short: pass
rate drops by 21\%, averaged across both models.


Yet most prior work evaluates skills \emph{one at a time}: how to author a
skill, retrieve one, or measure its marginal contribution over a no-skill
baseline~\citep{qin2023toolllm, patil2023gorilla, hao2023toolkengpt}.
Evidence from the more mature tool-use literature suggests this single-skill
view is dangerously incomplete. Tool selection accuracy drops from above
90\% with fewer than 30 candidates to 13.6\% with
11{,}100~\citep{gan2025ragmcp}, and in-context selection ``could completely
fail'' with numerous options~\citep{hao2023toolkengpt}. Skills amplify
these failure modes: they rely on natural-language descriptions with
broader, more overlapping scope than typed tool schemas, and a wrong skill
injects plausible but incorrect reasoning into the agent's context where it
is treated as authoritative~\citep{xu2026agentskills}---whereas a wrong
tool call is more likely to produce a detectable error (type mismatch,
missing parameter, API error code) due to its structured schema.


\textbf{Motivations.}
Despite growing awareness---recent work has documented degradation under progressively realistic skill settings~\citep{liu2026skillsinthewild} and shown that compressing skill content can yield a ``less-is-more''
effect~\citep{gao2026skillreducer}. As shown in Figure~\ref{fig:rq1-delta}, we evaluate on SkillsBench~\citep{li2026skillsbench} and observe monotonic pass rate drop when we load progressively larger skill libraries compared to only loading known-helpful skills (i.e., oracle skills). Our leading hypothesis is that \emph{certain skills in the library shadow the oracle skill}. For example:
\begin{tcolorbox}[colback=gray!10, colframe=gray!10, boxrule=0pt]               
  In the \texttt{mario-coin-counting} task, the library includes a skill named \texttt{video-frame-extraction}, whose description superficially matches the query better than the provided helpful skill. As a result, the agent invokes this wrong one in ALL $26$ trajectories. 
  \end{tcolorbox}

However, no prior work has formalized the library-induced degradation, nor has tried to decompose it into separable causal mechanisms through controlled experimentation. The challenge is that an agent trajectory---from receiving the task query, through skill selection and invocation, to producing a final output---is largely a black box: we observe the task outcome but not the intermediate decisions that produced it.
Without observing full intermediate states, we cannot attribute degradation to specific failure mechanisms.
   
\begin{figure}[!ht]         
    \centering
    \includegraphics[width=0.75\textwidth]{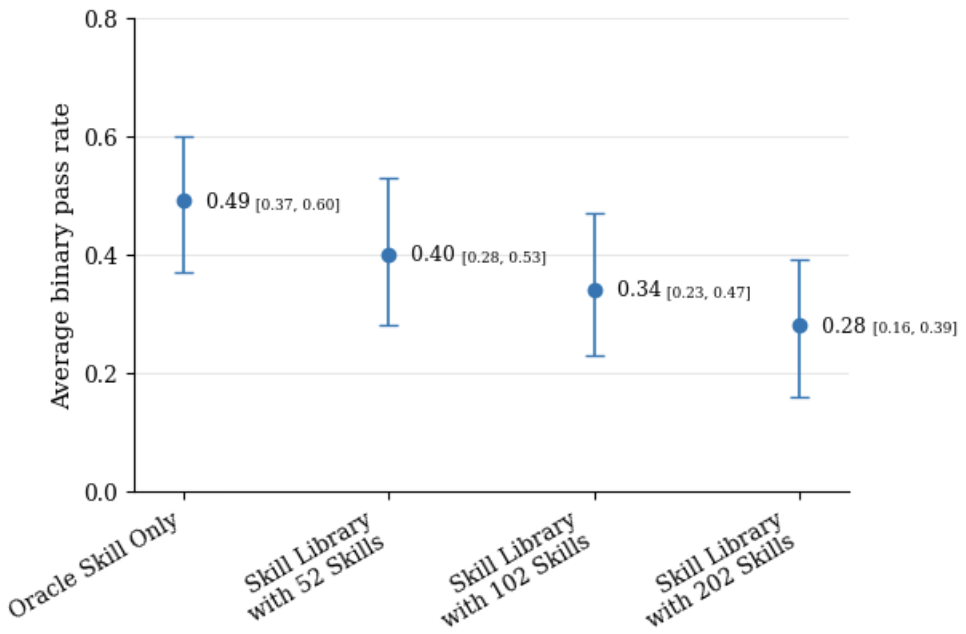}                              
    \caption{Pass-rate drop when expanding from oracle skills to larger libraries, averaged across $38$ (task, model) pairs, with $95\%$ percentile-bootstrap CI.}\label{fig:rq1-delta}
\end{figure}   



\textbf{Contributions.}
We propose to use \emph{skill invocations} as one such observable intermediate state.
We formalize the trajectory structure to define
invocation patterns, which enables an exact decomposition of library-induced degradation into two
separable effects---\emph{skill shadowing} and \emph{context
overhead}---with upper bounds derived on each individual effect (\S\ref{sec:shadowing}).
Our experiments show that skill shadowing dominates, accounting for up to 68\% of degradation and serving as the only statistically significant effect, while context overhead yields consistently reasonable but insignificant point estimates, remaining indistinguishable from noise (\S\ref{sec:experiments}).

\section{Formulation}\label{sec:setup}

We study \emph{skill-augmented} agents---LLM agents that select and load task-specific instruction sets from a shared library.
Unlike tool-augmented agents~\citep{qin2023toollearning,schick2023toolformer}, where invoking a tool executes code and returns a computed result, invoking a \emph{skill} injects instructional text into the agent's context window; the agent then drives task execution itself using base tools, guided by the injected instructions.

\subsection{Skill-Augmented Agents}

\begin{definition}[Skill]\label{def:skill}
A skill is a triple $S_i = (n_i, d_i, b_i)$ consisting of a name~$n_i$, a natural-language description~$d_i$, and a body~$b_i$ (the full instructions and examples that guide task execution).
The pair $(n_i, d_i)$ is the skill's \emph{interface}---the only information visible to the agent at selection time.
The body~$b_i$ is withheld until the skill is invoked.
\end{definition}

\begin{definition}[Skill library]\label{def:library}
A skill library is a finite set $\mathcal{S} = \{S_1, \ldots, S_\ell\}$.
\end{definition}

Our formulation adapts the agentic-skill tuple of \citet{sok2026skills}, who define a skill by its applicability condition, policy, termination condition, and callable interface.
Context-injection skills delegate execution entirely to the host agent, so the applicability condition, policy, and termination condition reduce to properties of the agent rather than the skill;
What remains is the callable interface (our $(n_i, d_i)$) and the instructional payload (our $b_i$).

\begin{definition}[Agent]\label{def:agent}
An agent $\mathcal{A}_\theta$ is a large language model with parameters~$\theta$ that, at each turn, selects an action conditioned on the current observation, prior interaction history, task~$q$, and the skill description $\{(n_i, d_i)\}_{i=1}^{\ell}$.
An action is one of: a \emph{skill invocation} $\texttt{Skill}(n_i)$, a \emph{base-tool call} (e.g., file read, shell command), or a \emph{termination signal}.
\end{definition}

\begin{definition}[Agent context]\label{def:context}
The \emph{agent context} is the complete input the agent receives before acting:
\begin{equation}\label{eq:context}
  \mathcal{C}(q, \mathcal{S}) = \bigl(q,\; \mathcal{D}(\mathcal{S})\bigr), \quad \mathcal{D}(\mathcal{S}) = \{(n_i, d_i)\}_{i=1}^{\ell}.
\end{equation}
The agent context is fully determined by the task and library, and is the only input that varies across experiments (\S\ref{sec:experiments}).
\end{definition}

This adapts the controller formulation of \citet{qin2023toollearning}, who decompose the LLM's decision as $p_\theta(a_t \mid x_t, H_t, q)$ over perceiver output~$x_t$, history~$H_t$, and query~$q$.
In our setting, the perceiver output is replaced by the skill description~$\mathcal{D}(\mathcal{S})$, and the history~$H_t$ is unobservable---we treat the trajectory as a black box from which only the invocation set~$\mathcal{I}$ and the final output are recorded.

\subsection{Formalization of the Trajectory}\label{sec:lifecycle}

\begin{definition}[Task]\label{def:task}
A task is a triple $\mathcal{T} = (q, r, v)$ consisting of a query~$q$ (a natural-language description of the objective), a reference solution~$r$ (a known-correct implementation that demonstrates solvability), and a verifier~$v$ (a deterministic test suite of independently evaluable subtasks).
\end{definition}

This follows the task structure of \citet{li2026skillsbench}, simplified by folding the environment into~$v$.
The agent observes only~$q$; both $r$ and~$v$ are withheld.
Since each task is uniquely identified by its query, we index by~$q$ for brevity throughout. Given $\mathcal{T} = (q, r, v)$ and $\mathcal{S}$, a trajectory proceeds in three phases:

\textbf{Phase~1: Context loading.}
The agent receives $q$ and $\mathcal{D}(\mathcal{S})$.
Skill bodies are withheld.

\textbf{Phase~2: Skill invocation to fulfill the task query.}
The agent acts for a bounded number of turns (or time).
At each turn, it emits an action: 2-(1) If the action is $\texttt{Skill}(n_i)$, the body~$b_i$ is appended to the agent's context (i.e., skill invocation); 2-(2) If the action is a base-tool call, the system executes it and returns the result. 

The agent may invoke multiple skills during a single trajectory.
Let $\mathcal{I} \subseteq \mathcal{S}$ denote the set of distinct skills invoked during the trajectory.
The intermediate actions, observations, and internal context accumulated during Phase~2 are difficult to observe.
We observe only the agent context~$\mathcal{C}(q, \mathcal{S})$, the invocation set~$\mathcal{I}$, and the final output evaluated by~$v$.

\textbf{Phase~3: Output.}
The trajectory terminates when the agent signals completion or reaches the turn (or time) limit.
The agent produces a final artifact (e.g., modified files, generated code) and the trajectory ends.

\subsection{Output Verification}\label{sec:metrics}

The verifier~$v$, external to the agent trajectory, evaluates the final environment state; the agent cannot observe or interact with~$v$.
We define the task completion random variable
\begin{equation}\label{eq:task-completion}
  T_q = \frac{\#\{\text{passed subtasks verified by } v\}}{m} \in \left\{0,\, \tfrac{1}{m},\, \ldots,\, 1\right\}.
\end{equation}
The \textit{binary pass rate} is the probability that all subtasks pass:
\begin{equation}\label{eq:pass-rate}
  p(q, \mathcal{A}_\theta, \mathcal{S}) = \Pr\!\left[T_q = 1 \mid q,\, \mathcal{A}_\theta,\, \mathcal{S}\right].
\end{equation}
The \textit{fractional pass rate} is:
\begin{equation}\label{eq:fractional-pass-rate}
  \bar{p}(q, \mathcal{A}_\theta, \mathcal{S}) = \mathbb{E}\!\left[T_q \mid q,\, \mathcal{A}_\theta,\, \mathcal{S}\right].
\end{equation}
When $m = 1$, $T_q \in \{0, 1\}$ and the two metrics coincide.
Since the agent~$\mathcal{A}_\theta$ is fixed throughout our experiments, we write $p(q, \mathcal{S})$ for brevity.

\section{Library-Induced Degradation: Decomposition and Upper Bounds}\label{sec:shadowing}

We formalize the two degradation mechanisms studied in this paper.
\emph{Skill shadowing} occurs when distractor skills displace oracle skills during selection, reducing the probability that the agent invokes a correct skill despite one being available.
\emph{Context overhead} occurs when the enlarged agent context degrades execution quality even when the agent selects correctly.
We define oracle and distractor skills, classify trajectories by their invocation pattern, and derive an exact decomposition of the pass-rate drop into these two effects.

\begin{definition}[Oracle skill set]\label{def:oracle}
For a task~$q$ and a skill-delta threshold~$\tau > 0$, the oracle skill set is the subset of skills whose marginal contribution over the no-skill baseline meets or exceeds~$\tau$:
\begin{equation}\label{eq:oracle}
  \mathcal{S}^\star(q) = \left\{S_i \in \mathcal{S} \;\middle|\; p(q, \{S_i\}) - p(q, \varnothing) \geq \tau\right\}.
\end{equation}
The oracle set is determined empirically: each skill is tested in isolation against~$q$, and only those whose pass-rate improvement over the no-skill baseline meets the threshold are included.
The set may contain multiple skills (when several skills independently improve performance on~$q$) or be empty (when no single skill provides sufficient benefit).
The value of $\tau$ and the empirical oracle-set sizes are reported in \S\ref{sec:experiments}.
\end{definition}

\begin{definition}[Distractor skills]\label{def:distractor}
For a task~$q$, all non-oracle skills $\mathcal{S} \setminus \mathcal{S}^\star(q)$ are distractors.
Distractors are not adversarially constructed; they are simply the skills in the library that do not meet the oracle threshold for~$q$.
\end{definition}

\begin{definition}[Invocation events]\label{def:invocation-events}
For a task~$q$ with $\mathcal{S}^\star(q) \neq \varnothing$, recall that a trajectory produces an invocation set~$\mathcal{I}$.
Every trajectory falls into one of three mutually exclusive events:

\begin{enumerate}
  \item[$\textsc{n}$.]
    \textbf{No skill invoked} ($\mathcal{I} = \varnothing$).
    The agent attempts the task with base tools alone, or fails to act.
  \item[$\textsc{m}$.]
    \textbf{Mixed invocation} ($\mathcal{I} \neq \varnothing$ and $\mathcal{I} \not\subseteq \mathcal{S}^\star(q)$).
    The agent invoked at least one skill, but not exclusively oracle skills---either no oracle was invoked at all, or an oracle was invoked alongside distractors.
  \item[$\textsc{o}$.]
    \textbf{Oracle-only invocation} ($\varnothing \neq \mathcal{I} \subseteq \mathcal{S}^\star(q)$).
    Every skill the agent invoked was an oracle skill.
\end{enumerate}
\end{definition}

\textbf{Library-induced pass-rate drop.}
To quantify the cost of library expansion, we compare the pass rate when only the oracle skill set is available (selection is trivial) against the pass rate under the full library:

\begin{definition}[Library-induced pass-rate drop]\label{def:delta}
For a task~$q$ with $\mathcal{S}^\star(q) \neq \varnothing$, the library-induced pass-rate drop is:
\begin{equation}\label{eq:delta}
  \Delta(q, \mathcal{S})
  = p\!\left(q,\, \mathcal{S}^\star(q)\right)
  \;-\;
  p(q,\, \mathcal{S}).
\end{equation}
\end{definition}

\textbf{Decomposition by invocation event.}
For tasks with $\mathcal{S}^\star(q) \neq \varnothing$, the three invocation events (Definition~\ref{def:invocation-events}) partition the trajectory space.
By the law of total probability:
\begin{equation}\label{eq:three-way}
  p(q, \mathcal{S})
  = \pi_{\textsc{n}}\,\rho_{\textsc{n}}
  + \pi_{\textsc{m}}\,\rho_{\textsc{m}}
  + \pi_{\textsc{o}}\,\rho_{\textsc{o}},
\end{equation}
where $\pi_{\textsc{n}}, \pi_{\textsc{m}}, \pi_{\textsc{o}}$ are the probabilities of events \textsc{n}, \textsc{m}, \textsc{o} given $(q, \mathcal{S})$, and $\rho_{\textsc{n}}, \rho_{\textsc{m}}, \rho_{\textsc{o}}$ are the corresponding conditional pass rates.
The probabilities of the three events sum to one, i.e., $\pi_{\textsc{n}} + \pi_{\textsc{m}} + \pi_{\textsc{o}} = 1$.

The same decomposition applies when the library contains only oracle skills, i.e., $\mathcal{S} = \mathcal{S}^\star(q)$.
Since no distractors are present, \textsc{m} is impossible, leaving only \textsc{n} and \textsc{o}.
We denote the event probabilities and conditional pass rates under $\mathcal{S} = \mathcal{S}^\star(q)$ with a star:
\begin{equation}\label{eq:oracle-decomp}
  p(q, \mathcal{S}^\star)
  = \pi_{\textsc{n}}^\star\,\rho_{\textsc{n}}^\star
  + \pi_{\textsc{o}}^\star\,\rho_{\textsc{o}}^\star,
\end{equation}
where $\pi_{\textsc{n}}^\star + \pi_{\textsc{o}}^\star = 1$.
Substituting Eqs.~\ref{eq:three-way} and~\ref{eq:oracle-decomp} into Eq.~\ref{eq:delta}:
\begin{equation}\label{eq:delta-decomp}
  \Delta(q, \mathcal{S})
  = \bigl(\pi_{\textsc{n}}^\star\,\rho_{\textsc{n}}^\star
  + \pi_{\textsc{o}}^\star\,\rho_{\textsc{o}}^\star\bigr)
  \;-\;
  \bigl(\pi_{\textsc{n}}\,\rho_{\textsc{n}}
  + \pi_{\textsc{m}}\,\rho_{\textsc{m}}
  + \pi_{\textsc{o}}\,\rho_{\textsc{o}}\bigr).
\end{equation}
Adding and subtracting $\pi_{\textsc{n}}^\star\,\rho_{\textsc{n}} + \pi_{\textsc{o}}^\star\,\rho_{\textsc{o}}$ yields:
\begin{align}\label{eq:two-effect}
  \Delta(q, \mathcal{S})
  &= \underbrace{%
      \pi_{\textsc{n}}^\star\bigl(\rho_{\textsc{n}}^\star - \rho_{\textsc{n}}\bigr)
      + \pi_{\textsc{o}}^\star\bigl(\rho_{\textsc{o}}^\star - \rho_{\textsc{o}}\bigr)
    }_{\displaystyle \Delta_{\mathrm{ctx}}} 
  \quad+\quad
  \underbrace{%
      \bigl(\pi_{\textsc{n}}^\star - \pi_{\textsc{n}}\bigr)\,\rho_{\textsc{n}}
      + \bigl(\pi_{\textsc{o}}^\star - \pi_{\textsc{o}}\bigr)\,\rho_{\textsc{o}}
      - \pi_{\textsc{m}}\,\rho_{\textsc{m}}
    }_{\displaystyle \Delta_{\mathrm{shd}}}.
\end{align}

\textbf{Context overhead effect} ($\Delta_{\mathrm{ctx}}$).
Event probabilities are held at their $\mathcal{S}^\star$ values; only the conditional pass rates change.
A positive $\rho_{\textsc{o}}^\star - \rho_{\textsc{o}}$ indicates that the full-library context degrades execution even when the agent selects the same skills.
The sole difference between conditions is the agent context (Definition~\ref{def:context}): $\mathcal{D}(\mathcal{S}^\star)$ contains $|\mathcal{S}^\star(q)|$ entries, while $\mathcal{D}(\mathcal{S})$ contains~$\ell$.

\textbf{Skill shadowing effect} ($\Delta_{\mathrm{shd}}$).
By analogy with variable shadowing in programming, distractor skills whose descriptions overlap with oracle skills can hide the correct choice, diverting the agent toward wrong or no skills.
Conditional pass rates are held at their full-library values; only the event probabilities change.
Under the full library, probability mass shifts from \textsc{o} into \textsc{m} and \textsc{n}, reducing performance.
Since $\pi_{\textsc{n}}^\star + \pi_{\textsc{o}}^\star = 1 = \pi_{\textsc{n}} + \pi_{\textsc{m}} + \pi_{\textsc{o}}$, we have $\pi_{\textsc{m}} = (\pi_{\textsc{n}}^\star - \pi_{\textsc{n}}) + (\pi_{\textsc{o}}^\star - \pi_{\textsc{o}})$, and the skill shadowing effect simplifies to:
\begin{equation}\label{eq:delta-shd-simplified}
  \Delta_{\mathrm{shd}}
  = \bigl(\pi_{\textsc{n}}^\star - \pi_{\textsc{n}}\bigr)
    \bigl(\rho_{\textsc{n}} - \rho_{\textsc{m}}\bigr)
  + \bigl(\pi_{\textsc{o}}^\star - \pi_{\textsc{o}}\bigr)
    \bigl(\rho_{\textsc{o}} - \rho_{\textsc{m}}\bigr).
\end{equation}

\begin{assumption}\label{asn:ctx}
For each invocation event $E \in \{\textsc{n}, \textsc{o}\}$, the conditional pass rate under $\mathcal{S} = \mathcal{S}^\star(q)$ is at least as high as under the full library:
\begin{equation}
  \rho_E^\star \geq \rho_E.
\end{equation}
\end{assumption}

\begin{lemma}[Upper bound on context overhead]\label{lem:ctx-bound}
Under Assumption~\ref{asn:ctx}, $\Delta_{\mathrm{ctx}} \geq 0$ and
\begin{equation}\label{eq:ctx-bound}
  \Delta_{\mathrm{ctx}}
  \leq \max\!\bigl(\rho_{\textsc{n}}^\star - \rho_{\textsc{n}},\;
                    \rho_{\textsc{o}}^\star - \rho_{\textsc{o}}\bigr).
\end{equation}
\end{lemma}
\begin{proof}
Both terms in $\Delta_{\mathrm{ctx}}$ are non-negative by Assumption~\ref{asn:ctx}.
Since $\pi_{\textsc{n}}^\star + \pi_{\textsc{o}}^\star = 1$, $\Delta_{\mathrm{ctx}}$ is a convex combination, bounded by the maximum of its components.
\end{proof}

\begin{lemma}[Upper bound on skill shadowing]\label{lem:shd-bound}
The skill shadowing effect satisfies
\begin{align}\label{eq:shd-bound}
  \bigl|\Delta_{\mathrm{shd}}\bigr|
  &\leq \bigl(|\pi_{\textsc{n}}^\star - \pi_{\textsc{n}}|
       + |\pi_{\textsc{o}}^\star - \pi_{\textsc{o}}|\bigr)
  \;\cdot\;
  \max\!\bigl(|\rho_{\textsc{n}} - \rho_{\textsc{m}}|,\;
              |\rho_{\textsc{o}} - \rho_{\textsc{m}}|\bigr) \notag\\
  &\leq |\pi_{\textsc{n}}^\star - \pi_{\textsc{n}}|
       + |\pi_{\textsc{o}}^\star - \pi_{\textsc{o}}|.
\end{align}
\end{lemma}
\begin{proof}
Apply the triangle inequality to Eq.~\ref{eq:delta-shd-simplified} and factor out the maximum absolute conditional pass-rate difference.
The second inequality follows because $|\rho_E - \rho_{\textsc{m}}| \leq 1$.
\end{proof}


\textbf{Interpretation.}
The two bounds serve as structural diagnostics: if the full-library context barely changes execution quality within each invocation event, the bound---and therefore $\Delta_{\mathrm{ctx}}$ itself---must be small, regardless of sample size.
A small bound thus disambiguates ``undetectable due to noise'' from ``genuinely small.''
The shadowing bound factors into two interpretable components: \emph{selection instability} (how much the invocation-event distribution shifts under the full library) and \emph{selection stakes} (how much pass rate differs across events).
Shadowing can be large only when both factors are large: the agent's selection must shift \emph{and} the shift must move mass toward lower-performing events.
We evaluate these bounds numerically in Appendix~\ref{app:bounds}.

\section{Experiments}\label{sec:experiments}

We design controlled experiments to answer the remaining question: what drives the library-induced degradation?

\textbf{Dataset.}
We evaluate on SkillsBench~\citep{li2026skillsbench}, a benchmark of $88$ tasks spanning document processing, finance, scientific computing, data analysis, etc.
Each task has a corresponding skill that could be helpful. 
We include only tasks whose corresponding skill yields at least a $4\%$ pass-rate uplift over the no-skill baseline, giving us $38$ (task, model) pairs with oracle skills ($21$ tasks from Haiku~4.5, and $17$ from Sonnet~4.6).
For each pair we evaluate under a baseline where the library
contains only oracle skills ($\mathcal{S}{=}\mathcal{S}^\star$) and three expanded libraries of total size $|\mathcal{S}| \in \{52, 102, 202\}$.
The final dataset comprises $2{,}545$ trajectories.
Further details, e.g., task selection, full trajectory counts, library construction, and bootstrap, are deferred to Appendix~\ref{app:protocol}.

\textbf{Decomposition and scaling of the pass-rate drop.}
For each library size we classify trajectories into the three invocation
events (Definition~\ref{def:invocation-events}) and report $\Delta_{\mathrm{ctx}}$,
$\Delta_{\mathrm{shd}}$, and $\Delta$ (Eq.~\ref{eq:two-effect})
in Table~\ref{tab:rq2-decomp}:
Both $\Delta_{\mathrm{ctx}}$ and
$\Delta_{\mathrm{shd}}$ grow with library size, with
$\Delta_{\mathrm{shd}}$ accounting for an increasing share of
$\Delta$.
The behavioral shifts underlying this scaling are visible in
Table~\ref{tab:rq2-invoc}, which reports the fraction of trajectories
in each of four invocation patterns---``oracle only''
(event~\textsc{o}), ``mixed -- oracle invoked'' and ``mixed --
oracle not invoked'' (sub-cases of event~\textsc{m}), and ``no skill
invoked'' (event~\textsc{n}).
As the library grows, the event~\textsc{o} share collapses, with mass
migrating into events~\textsc{n} and~\textsc{m}.
Two complementary failure modes are visible:
\emph{abandonment} (mass into event~\textsc{n}) and
\emph{wrong selection} (mass into event~\textsc{m}).
Additional shadowing examples appear in Appendix~\ref{app:shadowing-examples} and empirical studies of invocation pattern shift can be found in Appendix~\ref{app:invocation-dist}.

\vspace{-0.1in}
\begin{table}[!ht]
  \centering
  \footnotesize
  \setlength{\tabcolsep}{10pt}
  \caption{Decomposition of $\Delta$
  (Definition~\ref{def:delta}) into $\Delta_{\mathrm{ctx}}$ and
  $\Delta_{\mathrm{shd}}$ via Eq.~\ref{eq:two-effect}, averaged across
  both models, with $95\%$ two-stage clustered-bootstrap CIs.
  $\Delta_{\mathrm{ctx}} + \Delta_{\mathrm{shd}} = \Delta$ exactly.
  \textbf{Bold} marks CIs that exclude zero.}
  \label{tab:rq2-decomp}
  \vspace{0.1in}
  \begin{tabular}{l c c c}
    \toprule
    $|\mathcal{S}|$ & $\Delta_{\mathrm{ctx}}$ & $\Delta_{\mathrm{shd}}$ & $\Delta$ \\
    \midrule
    $52$  & $.06\;[-.11, .22]$ & $.03\;[-.01, .07]$ & $\mathbf{.08\;[.02, .15]}$ \\
    $102$ & $.06\;[-.11, .23]$ & $\mathbf{.08\;[.02, .15]}$ & $\mathbf{.14\;[.09, .19]}$ \\
    $202$ & $.07\;[-.13, .25]$ & $\mathbf{.14\;[.06, .26]}$ & $\mathbf{.21\;[.15, .27]}$ \\
    \bottomrule
  \end{tabular}
\end{table}

\textbf{Which effect dominates the drop?}
$\Delta_{\mathrm{shd}}$ is the only component whose CI excludes zero
and accounts for the majority of $\Delta$ at the largest library.
$\Delta_{\mathrm{ctx}}$ does not separate from zero at any library
size, although its point estimate is consistently positive.
Within event~\textsc{o}---where the agent selected
correctly---the binary pass rate still falls as the library grows
(see Tables~\ref{tab:rq4-2-binary} and \ref{tab:rq4-2-fraction} in Appendix~\ref{app:finer-partition}),
suggesting context overhead is real but our sample size is too small to confirm it.
Based on the empirical observation, we conclude that shadowing is the dominant effect of the drop, and the context-overhead point estimates as suggestive only.
The theoretical bounds of Lemmas~\ref{lem:ctx-bound} and~\ref{lem:shd-bound} corroborate this conclusion: the context
overhead bound remains small and indistinguishable from zero, while the shadowing bound grows sharply and is significantly positive throughout; see Appendix~\ref{app:bounds} for complete details.

\vspace{-0.1in}
\begin{table}[!ht]
  \centering
  \footnotesize
  \setlength{\tabcolsep}{4pt}
  \caption{Trajectory-share distribution across the four invocation
  patterns (refinement of Definition~\ref{def:invocation-events}'s partition; see Appendix~\ref{app:finer-partition} for details).}
  \label{tab:rq2-invoc}
  \vspace{0.1in}
  \begin{tabular}{c rrrr r}
    \toprule
     & \makecell{oracle\\only} & \makecell{mixed --\\oracle invoked} & \makecell{mixed --\\oracle not invoked} & \makecell{no skill\\invoked} & $n$ \\
    \midrule
    $\mathcal{S} = \mathcal{S}^\star$   & $88.0\%$ & $0.0\%$ & $0.0\%$ & $12.0\%$ & $552$ \\
    $|\mathcal{S}| = 52$  & $80.3\%$ & $0.9\%$ & $2.8\%$ & $16.1\%$ & $542$ \\
    $|\mathcal{S}| = 102$ & $64.4\%$ & $1.4\%$ & $5.9\%$ & $28.3\%$ & $988$ \\
    $|\mathcal{S}| = 202$ & $52.6\%$ & $1.6\%$ & $7.3\%$ & $38.5\%$ & $494$ \\
    \bottomrule
  \end{tabular}
\end{table}

\section{Discussion}\label{sec:discussion}

Our results point to skill shadowing as the primary bottleneck when scaling a skill library: once semantically similar skills compete for selection, accuracy degrades regardless of the underlying model.
A natural next step is to attack this bottleneck directly through retrieval-based pre-filtering~\citep{gan2025ragmcp, wang2025toolgen}, description disambiguation~\citep{gao2026skillreducer}, learned routing~\citep{jia2026autotool}, or other techniques. In addition, our study covers $202$ skills on two models from a single provider, a scale at which context overhead remains modest.
Whether these patterns---and the mitigations above---hold as libraries grow by an order of magnitude or span diverse architectures remains an open question.

\section*{Acknowledgment}

We thank Zubair Anwar, Juanyan Li and Martin Valdez-Vivas for their valuable feedback and discussions throughout this work. Moreover, we truly appreciate the Data Science Leaders Sam Shah, Feng Pan, Zubair Anwar, and Divy Menghani for their support.

\bibliographystyle{plainnat}
\bibliography{corrected_ref_v1}

\appendices

\addcontentsline{toc}{section}{Appendix} 
\part{\centering \LARGE Appendix} 

\topskip0pt

\parttoc 


\section{Experiment Configurations}\label{app:protocol}

Definition~\ref{def:oracle} defines the oracle set at the individual-skill level, but in practice skills are often bundled together: multiple skills may jointly serve as the oracle set for a given task. In Table~\ref{tab:shadowing-examples}, \texttt{trend-anomaly-causal-inference} is a good example.

\begin{table*}[!ht]
  \centering
  \scriptsize
  \setlength{\tabcolsep}{0.1pt}
  \renewcommand{\arraystretch}{1.05}
  \caption{Frequent shadowing pairs at $|\mathcal{S}|{=}202$
  ($494$ trajectories across the $38$ matched (task, model) pairs).
  $\#_{\text{shd traj}}$ is the count of trajectories in which the agent
  first-picks the distractor skill.}
  \label{tab:shadowing-examples}
  \vspace{0.1in}
  \begin{tabular}{@{}>{\raggedright\arraybackslash}p{4cm} >{\raggedright\arraybackslash}p{4cm} >{\raggedright\arraybackslash}p{3.5cm} >{\raggedright\arraybackslash}p{4.5cm} r@{}}
    \toprule
    Primary task & Oracle bundle & distractor skill & Origin task(s) of distractor skill & $\#_{\text{shd traj}}$ \\
    \midrule
    \texttt{mario-coin-counting}
      & \makecell[l]{\texttt{ffmpeg},\\\texttt{image\_editing},\\\texttt{object\_counter}}
      & \texttt{video-frame-extraction}
      & \makecell[l]{\texttt{jpg-ocr-stat},\\\texttt{pedestrian-traffic-counting}}
      & $51$ \\
    \addlinespace
    \texttt{trend-anomaly-} \texttt{causal-inference} 
      & \makecell[l]{\texttt{data\_cleaning},\\\texttt{did\_causal\_analysis},\\\texttt{feature\_engineering},\\\texttt{time\_series\_anomaly\_detection}}
      & \texttt{senior-data-scientist}
      & \texttt{powerlifting-coef-calc}
      & $37$ \\
    \addlinespace
    \texttt{manufacturing-} \texttt{fjsp-optimization}
      & \makecell[l]{\texttt{fjsp-baseline-repair-}\\\texttt{with-downtime-and-policy}}
      & \makecell[l]{\texttt{reflow\_machine\_}\\\texttt{maintenance\_guidance}}
      & \texttt{manufacturing-equipment-} \texttt{maintenance}
      & $14$ \\
    \addlinespace
    \texttt{econ-detrending-} \texttt{correlation}
      & \texttt{timeseries-detrending}
      & \texttt{csv-processing}
      & \texttt{adaptive-cruise-control}
      & $4$ \\
    \addlinespace
    \texttt{energy-market-pricing}
      & \makecell[l]{\texttt{dc-power-flow},\\\texttt{economic-dispatch},\\\texttt{locational-marginal-prices},\\\texttt{power-flow-data}}
      & \texttt{casadi-ipopt-nlp}
      & \texttt{energy-ac-optimal-power-flow}
      & $3$ \\
    \bottomrule
  \end{tabular}
\end{table*}

In SkillsBench~\citep{li2026skillsbench}, we have \emph{skill bundle} (or just bundle for brevity) authored by task creators---they are hand-crafted set of skills that the task author considers helpful for solving the given task. We call this set the task's \emph{authored primary skill bundle}, or just \emph{authored (skill) bundle} for brevity. Bundles range from $1$ to $7$ skills (median~$2$, mean~$2.7$); they are provided in the benchmark, not by this work.

\subsection{Tasks included in our experiments}
We only include a (task, model) pair into our experiment if every individual skill in the authored bundle clears the empirical threshold $\tau = 0.04$ per our Definition~\ref{def:oracle}, so that we can ``break the bundle structure'' and ensure each individual skill helps solve the task. The rationale is simple: if the skill does not help, one probably could not observe performance degradation even with more skills in the library. 
This filtering leaves us $38$ (task, model) pairs: $21$ Haiku and $17$ Sonnet. 
From now on, the authored bundle serves as the oracle skill set $\mathcal{S}^\star(q)$, and we use ``oracle skill set'' and ``oracle bundle'' interchangeably below.
Our experiments consist of $2{,}545$ trajectories:
with $303/297/544/273$ from Haiku and
$233/233/441/221$ from Sonnet for $\mathcal{S} = \mathcal{S}^\star$, and
$|\mathcal{S}| \in \{52, 102, 202\}$, respectively.

\subsection{Library construction}
The largest library ($|\mathcal{S}|{=}202$) is the union of all $88$
authored bundles after byte-identical de-duplication and
numeric-suffix disambiguation for same-name skills with divergent bodies.

\subsection{Bootstrap}
Confidence intervals come from $95\%$ percentile of the bootstrap distribution.
For the decomposition (Eq.~\ref{eq:two-effect}): at each bootstrap trial, we first resample (task, model) pairs with replacement, then resample trajectories within each sampled pair.

\subsection{Examples with skill shadowing effect}\label{app:shadowing-examples}

Table~\ref{tab:shadowing-examples} lists the most frequent empirical examples with skill shadowing effect at $|\mathcal{S}|{=}202$, ranked by the number of
trajectories in which the distractor skill is the agent's first pick.
The first row is the \texttt{mario-coin-counting} example in our introduction
\S\ref{sec:intro}.

\section{Additional Experiments}

\subsection{Per-model pass rate drops}\label{app:per-model-rq1}

Table~\ref{tab:rq1-delta-full} presents the overall pass rate drops with a per-model breakdown, where we can observe both models degrade as the library grows.
Haiku's drop separates from zero already at $|\mathcal{S}|{=}52$;
Sonnet's only at $|\mathcal{S}|{=}202$.

\begin{table}[!ht]
  \centering
  \footnotesize
  \setlength{\tabcolsep}{20pt}
  \caption{Library-induced pass-rate drop $\Delta$
  (Definition~\ref{def:delta}) by model and library size, with
  $95\%$ percentile-bootstrap CI. \textbf{Bold} marks CIs excluding
  zero.}
  \label{tab:rq1-delta-full}
  \vspace{0.1in}
  \begin{tabular}{@{}l cc@{}}
    \toprule
    $|\mathcal{S}|$ & Haiku 4.5 & Sonnet 4.6 \\
    \midrule
    $52$  & $\mathbf{.10\;[.03, .18]}$ & $.06\;[-.03, .15]$ \\
    $102$ & $\mathbf{.22\;[.16, .29]}$ & $.05\;[-.03, .13]$ \\
    $202$ & $\mathbf{.26\;[.20, .33]}$ & $\mathbf{.15\;[.06, .24]}$ \\
    \bottomrule
  \end{tabular}
\end{table}

\begin{table*}[!ht]
  \centering
  \footnotesize
  \setlength{\tabcolsep}{10pt}
  \caption{Per-model decomposition of $\Delta$
  (Definition~\ref{def:delta}) into $\Delta_{\mathrm{ctx}}$ and
  $\Delta_{\mathrm{shd}}$ (Eq.~\ref{eq:two-effect}), with $95\%$
  joint two-stage clustered-bootstrap CIs
  ($2{,}000$ trials, random seed is~$42$).
  $\Delta_{\mathrm{ctx}}+\Delta_{\mathrm{shd}}=\Delta$ on every
  point estimate. \textbf{Bold} marks CIs excluding zero.
  Pooled estimates appear in Table~\ref{tab:rq2-decomp}.}
  \label{tab:rq2-decomp-full}
  \vspace{0.1in}
  \begin{tabular}{@{}l c c c c@{}}
    \toprule
    View & $|\mathcal{S}|$ & $\Delta_{\mathrm{ctx}}$ & $\Delta_{\mathrm{shd}}$ & $\Delta$ \\
    \midrule
    \multirow{3}{*}{Haiku 4.5}
      & $52$  & $.09\;[-.10, .27]$ & $.02\;[-.04, .07]$         & $.10\;[-.08, .28]$         \\
      & $102$ & $.15\;[-.04, .34]$ & $.07\;[.00, .16]$          & $\mathbf{.22\;[.05, .39]}$ \\
      & $202$ & $.09\;[-.18, .36]$ & $\mathbf{.17\;[.02, .37]}$ & $\mathbf{.26\;[.09, .43]}$ \\
    \addlinespace
    \multirow{3}{*}{Sonnet 4.6}
      & $52$  & $.05\;[-.22, .33]$ & $.01\;[-.07, .10]$         & $.06\;[-.21, .33]$         \\
      & $102$ & $.03\;[-.23, .30]$ & $.02\;[-.06, .10]$         & $.05\;[-.21, .31]$         \\
      & $202$ & $.11\;[-.16, .40]$ & $.04\;[-.04, .15]$         & $.15\;[-.11, .41]$         \\
    \bottomrule
  \end{tabular}
\end{table*}

Table~\ref{tab:rq2-decomp-full} presents the estimates on both effects with a per-model breakdown:
For Haiku, $\Delta_{\mathrm{shd}}$ dominates and its CI excludes zero at the largest library, whereas for Sonnet, $\Delta_{\mathrm{ctx}}$ matches or exceeds
$\Delta_{\mathrm{shd}}$ at every size, but we do not have necessary statistical power to distinguish from zero.

\subsection{Theoretical upper bounds}\label{app:bounds}

We substitute the empirical $(\pi^\star_E, \rho^\star_E,
\pi_E, \rho_E)$ from the decomposition into the right-hand sides of
Lemma~\ref{lem:ctx-bound} (Eq.~\ref{eq:ctx-bound}) and
Lemma~\ref{lem:shd-bound} (tighter form of Eq.~\ref{eq:shd-bound}).
CIs are obtained by recomputing each bound on every bootstrap trial and taking the proper percentiles.

\begin{table}[!h]
  \centering
  \footnotesize
  \setlength{\tabcolsep}{10pt}
  \caption{Upper bounds $\Delta_{\mathrm{ctx}}^{\sup}$ and
  $\Delta_{\mathrm{shd}}^{\sup}$ from Lemmas~\ref{lem:ctx-bound}
  and~\ref{lem:shd-bound} (tighter form), with $95\%$ joint two-stage
  clustered-bootstrap CIs from the same iterations as
  Table~\ref{tab:rq2-decomp-full}
  ($2{,}000$ trials, random seed is~$42$).
  \textbf{Bold} marks CIs excluding zero.}
  \label{tab:rq3-bounds-full}
  \vspace{0.1in}
  \begin{tabular}{@{}l cc cc cc@{}}
    \toprule
    & \multicolumn{2}{c}{Haiku 4.5} & \multicolumn{2}{c}{Sonnet 4.6} & \multicolumn{2}{c}{Pooled} \\
    \cmidrule(lr){2-3}\cmidrule(lr){4-5}\cmidrule(lr){6-7}
    $|\mathcal{S}|$ & $\Delta_{\mathrm{ctx}}^{\sup}$ & $\Delta_{\mathrm{shd}}^{\sup}$
                    & $\Delta_{\mathrm{ctx}}^{\sup}$ & $\Delta_{\mathrm{shd}}^{\sup}$
                    & $\Delta_{\mathrm{ctx}}^{\sup}$ & $\Delta_{\mathrm{shd}}^{\sup}$ \\
    \midrule
    $52$  & $.10\;[-.06, .36]$ & $\mathbf{.04\;[.01, .19]}$ & $.06\;[-.20, .33]$ & $.16\;[.00, .44]$         & $.08\;[-.08, .29]$ & $\mathbf{.06\;[.01, .18]}$ \\
    $102$ & $.17\;[.00, .41]$  & $\mathbf{.09\;[.02, .32]}$ & $.04\;[-.21, .30]$ & $\mathbf{.17\;[.01, .43]}$ & $.07\;[-.08, .31]$ & $\mathbf{.14\;[.04, .29]}$ \\
    $202$ & $.10\;[-.04, .43]$ & $\mathbf{.35\;[.08, .74]}$ & $.12\;[-.14, .41]$ & $\mathbf{.15\;[.01, .37]}$ & $.07\;[-.06, .34]$ & $\mathbf{.24\;[.09, .46]}$ \\
    \bottomrule
  \end{tabular}
\end{table}

Table~\ref{tab:rq3-bounds-full} reports the two upper bounds, where we denote $\Delta_{\mathrm{ctx}}^{\sup} = \max\!\bigl(\rho_{\textsc{n}}^\star - \rho_{\textsc{n}},\; \rho_{\textsc{o}}^\star - \rho_{\textsc{o}}\bigr)$, and
$\Delta_{\mathrm{shd}}^{\sup} = |\pi_{\textsc{n}}^\star - \pi_{\textsc{n}}|
       + |\pi_{\textsc{o}}^\star - \pi_{\textsc{o}}|$.
We observe $\Delta_{\mathrm{ctx}}^{\sup}$ remains small across all conditions, and no
CI separates from zero, confirming our finding that context overhead is not distinguishable from noise.
$\Delta_{\mathrm{shd}}^{\sup}$ grows sharply with library size and
excludes zero under nearly all scenarios.

\subsection{Pass rates for different invocation events}\label{app:finer-partition}

\textbf{A finer partition of the invocation events.}
Definition~\ref{def:invocation-events}'s mixed event~\textsc{m}
admits two types of events: the agent invoked at least one oracle skill
alongside distractors, or it invoked only distractors.
Combined with \textsc{n} and \textsc{o}, this gives a four-way partition: 
\begin{itemize}
    \item ``oracle only'' (\textsc{o});
    \item ``mixed -- oracle invoked'';
    \item ``mixed -- oracle not invoked'';
    \item ``no skill invoked'' (\textsc{n}).
\end{itemize}

We report the binary pass rate for all four types of invocation events in Table~\ref{tab:rq4-2-binary} and the fractional pass rate
in Table~\ref{tab:rq4-2-fraction}.
CIs are Wilson $95\%$ for binary and $95\%$ percentile bootstrap for fractional.

\begin{table}[!ht]
  \centering
  \footnotesize
  \setlength{\tabcolsep}{10pt}
  \caption{Binary pass rates for different invocation events and
  library sizes, with Wilson $95\%$ CIs. Cells with fewer than $5$
  valid trajectories are ``---''.}
  \label{tab:rq4-2-binary}
  \vspace{0.1in}
  \begin{tabular}{@{}l l c c c c@{}}
    \toprule
    Invocation pattern & View & $p\;(\mathcal{S}^\star)$ & $p\;(|\mathcal{S}|{=}52)$ & $p\;(|\mathcal{S}|{=}102)$ & $p\;(|\mathcal{S}|{=}202)$ \\
    \midrule
    \multirow{3}{*}{\makecell[l]{oracle\\only}}
      & Haiku 4.5  & $.43\;[.37, .50]$ & $.33\;[.28, .40]$    & $.26\;[.21, .32]$   & $.34\;[.24, .45]$ \\
      & Sonnet 4.6 & $.62\;[.56, .68]$ & $.57\;[.50, .63]$    & $.58\;[.53, .63]$   & $.51\;[.43, .58]$ \\
      & Pooled     & $.52\;[.48, .57]$ & $.44\;[.40, .49]$    & $.45\;[.41, .49]$   & $.45\;[.39, .52]$ \\
    \addlinespace
    \multirow{3}{*}{\makecell[l]{mixed --\\oracle invoked}}
      & Haiku 4.5  & ---               & ---                  & $.33\;[.12, .65]$   & ---               \\
      & Sonnet 4.6 & ---               & $.00\;[.00, .43]$    & $.20\;[.04, .62]$   & $.13\;[.02, .47]$ \\
      & Pooled     & ---               & $.00\;[.00, .43]$    & $.29\;[.12, .55]$   & $.13\;[.02, .47]$ \\
    \addlinespace
    \multirow{3}{*}{\makecell[l]{mixed --\\oracle not invoked}}
      & Haiku 4.5  & ---               & $.00\;[.00, .35]$    & $.05\;[.01, .22]$   & $.00\;[.00, .23]$ \\
      & Sonnet 4.6 & ---               & $.00\;[.00, .32]$    & $.14\;[.06, .29]$   & $.13\;[.05, .32]$ \\
      & Pooled     & ---               & $.00\;[.00, .20]$    & $.10\;[.05, .21]$   & $.08\;[.03, .22]$ \\
    \addlinespace
    \multirow{3}{*}{\makecell[l]{no skill\\invoked}}
      & Haiku 4.5  & $.12\;[.06, .24]$ & $.11\;[.06, .22]$    & $.07\;[.04, .11]$   & $.03\;[.02, .07]$ \\
      & Sonnet 4.6 & ---               & $1.00\;[.77, 1.00]$  & $1.00\;[.87, 1.00]$ & $.80\;[.49, .94]$ \\
      & Pooled     & $.12\;[.06, .24]$ & $.27\;[.18, .38]$    & $.16\;[.12, .20]$   & $.07\;[.04, .12]$ \\
    \bottomrule
  \end{tabular}
\end{table}

\begin{table}[!ht]
  \centering
  \footnotesize
  \setlength{\tabcolsep}{10pt}
  \caption{Fractional pass rate by invocation event type and library size, with $95\%$ percentile-bootstrap CIs. Cells with
  fewer than $5$ valid trajectories are ``---''.}
  \label{tab:rq4-2-fraction}
  \vspace{0.1in}
  \begin{tabular}{@{}l l c c c c@{}}
    \toprule
    Invocation pattern & View & $\bar p\;(\mathcal{S}^\star)$ & $\bar p\;(|\mathcal{S}|{=}52)$ & $\bar p\;(|\mathcal{S}|{=}102)$ & $\bar p\;(|\mathcal{S}|{=}202)$ \\
    \midrule
    \multirow{3}{*}{\makecell[l]{oracle\\only}}
      & Haiku 4.5  & $.83\;[.81, .86]$ & $.77\;[.73, .80]$    & $.77\;[.73, .80]$    & $.80\;[.73, .85]$  \\
      & Sonnet 4.6 & $.91\;[.88, .93]$ & $.89\;[.87, .92]$    & $.91\;[.89, .93]$    & $.80\;[.74, .84]$  \\
      & Pooled     & $.87\;[.85, .89]$ & $.83\;[.80, .85]$    & $.85\;[.83, .86]$    & $.80\;[.75, .83]$  \\
    \addlinespace
    \multirow{3}{*}{\makecell[l]{mixed --\\oracle invoked}}
      & Haiku 4.5  & ---               & ---                  & $.92\;[.87, .98]$    & ---                \\
      & Sonnet 4.6 & ---               & ---                  & ---                  & $.89\;[.81, .94]$  \\
      & Pooled     & ---               & ---                  & $.91\;[.85, .97]$    & $.89\;[.81, .94]$  \\
    \addlinespace
    \multirow{3}{*}{\makecell[l]{mixed --\\oracle not invoked}}
      & Haiku 4.5  & ---               & ---                  & $.70\;[.60, .79]$    & ---                \\
      & Sonnet 4.6 & ---               & ---                  & $.85\;[.79, .90]$    & $.64\;[.40, .87]$  \\
      & Pooled     & ---               & ---                  & $.80\;[.75, .85]$    & $.64\;[.40, .87]$  \\
    \addlinespace
    \multirow{3}{*}{\makecell[l]{no skill\\invoked}}
      & Haiku 4.5  & $.48\;[.36, .59]$ & $.54\;[.45, .64]$    & $.45\;[.40, .49]$    & $.42\;[.37, .48]$  \\
      & Sonnet 4.6 & ---               & $1.00\;[1.00, 1.00]$ & $1.00\;[1.00, 1.00]$ & $.80\;[.50, 1.00]$ \\
      & Pooled     & $.47\;[.36, .59]$ & $.62\;[.53, .71]$    & $.50\;[.45, .55]$    & $.44\;[.39, .50]$  \\
    \bottomrule
  \end{tabular}
\end{table}

\begin{itemize}
    \item \textbf{Context overhead is visible within the oracle-only bucket.} For ``oracle only'' (event~\textsc{o})---trajectories where the agent invoked only oracle skills---the binary pass rate falls monotonically with library size on both models. This is the empirical signature of context overhead: the same invocation pattern produces a lower pass rate under a larger in-context library.

    \item \textbf{Fractional pass rate falls less than binary.}
Within ``oracle only'', the fractional pass rate drops by less than
the binary pass rate at every library size (e.g.\ pooled
$\bar p$: $.87\!\to\!.80$ vs.\ pooled $p$: $.52\!\to\!.45$ at
$|\mathcal{S}|{=}202$).
Context overhead removes whole-task correctness more than partial
progress.

    \item \textbf{Mixed and no-skill buckets are sample-limited.}
The two ``mixed'' buckets and ``no skill invoked'' draw on small
samples in many conditions; their high-variance pass rates should be
read alongside the cell counts in Table~\ref{tab:rq5-invoc}.
\end{itemize}

\subsection{Invocation event distributions}\label{app:invocation-dist}

We call $\pi_{\textsc{m}}$ (Definition~\ref{def:invocation-events}) the empirical shadowing rate: the
fraction of trajectories in which the agent invokes at least one distractor skill.

Table~\ref{tab:rq4-shadowing} reports $\pi_{\textsc{m}}$ (with Wilson $95\%$ CIs) on the selected (task, model) pairs, and we find that
Sonnet accumulates more wrong-selection mass than Haiku at every library size; Haiku's $\pi_{\textsc{m}}$ stays low because its mass
shifts to ``no skill invoked'' events instead of selecting wrong skills as shown in Table~\ref{tab:rq5-invoc} later on.

\begin{table}[!h]
  \centering
  \footnotesize
  \setlength{\tabcolsep}{15pt}
  \caption{Empirical shadowing rate $\pi_{\textsc{m}}$
  (Definition~\ref{def:invocation-events}) by library size, with
  Wilson $95\%$ CIs. $\mathcal{S}\!=\!\mathcal{S}^\star$ row is omitted because $\pi_{\textsc{m}}\!=\!0$ by construction.}
  \label{tab:rq4-shadowing}
  \vspace{0.1in}
  \begin{tabular}{@{}l ccc@{}}
    \toprule
    $|\mathcal{S}|$ & Haiku 4.5 & Sonnet 4.6 & Pooled \\
    \midrule
    $52$  & $.02\;[.01, .05]$ & $.06\;[.03, .09]$ & $.04\;[.02, .06]$ \\
    $102$ & $.06\;[.04, .08]$ & $.09\;[.07, .12]$ & $.07\;[.06, .09]$ \\
    $202$ & $.05\;[.03, .08]$ & $.14\;[.10, .19]$ & $.09\;[.07, .12]$ \\
    \bottomrule
  \end{tabular}
\end{table}

\begin{table}[!ht]
  \centering
  \scriptsize
  \setlength{\tabcolsep}{22pt}
  \caption{Trajectory-share distribution across the four invocation
  patterns (refinement of Definition~\ref{def:invocation-events}'s
  \textsc{n}/\textsc{m}/\textsc{o} partition), by library size.}
  \label{tab:rq5-invoc}
  \vspace{0.1in}
  \begin{tabular}{@{}l c rrrr r@{}}
    \toprule
    & $|\mathcal{S}|$ & \makecell{oracle\\only} & \makecell{mixed ---\\oracle\\invoked} & \makecell{mixed ---\\oracle not\\invoked} & \makecell{no skill\\invoked} & $n$ \\
    \midrule
    \multirow{4}{*}{Haiku 4.5}
      & $\mathcal{S}^\star$ & $80.9\%$ & $0.0\%$ & $0.0\%$  & $19.1\%$           & $314$ \\
      & $52$                & $74.5\%$ & $0.0\%$ & $2.3\%$  & $23.2\%$           & $306$ \\
      & $102$               & $47.8\%$ & $1.6\%$ & $4.0\%$  & $46.5\%$           & $546$ \\
      & $202$               & $29.3\%$ & $0.0\%$ & $4.8\%$  & $65.9\%$           & $273$ \\
    \addlinespace
    \multirow{4}{*}{Sonnet 4.6}
      & $\mathcal{S}^\star$ & $97.5\%$ & $0.0\%$ & $0.0\%$  & $\phantom{0}2.5\%$ & $238$ \\
      & $52$                & $87.7\%$ & $2.1\%$ & $3.4\%$  & $\phantom{0}6.8\%$ & $236$ \\
      & $102$               & $84.8\%$ & $1.1\%$ & $8.1\%$  & $\phantom{0}5.9\%$ & $442$ \\
      & $202$               & $81.4\%$ & $3.6\%$ & $10.4\%$ & $\phantom{0}4.5\%$ & $221$ \\
    \addlinespace
    \multirow{4}{*}{Pooled}
      & $\mathcal{S}^\star$ & $88.0\%$ & $0.0\%$ & $0.0\%$  & $12.0\%$           & $552$ \\
      & $52$                & $80.3\%$ & $0.9\%$ & $2.8\%$  & $16.1\%$           & $542$ \\
      & $102$               & $64.4\%$ & $1.4\%$ & $5.9\%$  & $28.3\%$           & $988$ \\
      & $202$               & $52.6\%$ & $1.6\%$ & $7.3\%$  & $38.5\%$           & $494$ \\
    \bottomrule
  \end{tabular}
\end{table}

Table~\ref{tab:rq5-invoc} reports the per-bucket trajectory shares.
The two models exhibit qualitatively different transitions as the
library grows.
On Haiku, ``oracle only'' (event~\textsc{o}) collapses and the
displaced mass goes almost entirely into ``no skill invoked''
(event~\textsc{n}); the two mixed buckets stay small at every size.
Haiku's failure mode under expansion is therefore
\emph{abandonment}: rather than picking the wrong skill, the agent
picks no skill.
On Sonnet, event~\textsc{o} declines more gently, event~\textsc{n}
stays small, and the displaced mass goes into the two ``mixed''
sub-cases of event~\textsc{m}.
Sonnet's failure mode is \emph{wrong selection}: the agent keeps
invoking skills, but increasingly the wrong ones.
This per-model contrast underlies the decomposition pattern in
Table~\ref{tab:rq2-decomp}.

\end{document}